\documentclass[a4paper,UKenglish]{lipics-v2018}

\usepackage{microtype}

\usepackage{amssymb,amsmath,graphics,graphicx,color}

\usepackage{tikz} 
\usetikzlibrary{arrows.meta}

\usepackage{enumitem}

\newtheorem{observation}{\bf Observation}
\newtheorem{proposition}{\bf Proposition}

\newcommand{\mcds}{minimal connected dominating set}
\newcommand{\mcdss}{minimal connected dominating sets}

\newcommand*{\dprime}{^{\prime\prime}\mkern-1.2mu}

\usepackage{graphicx,color}

\usepackage{enumitem}

\newlength{\mg}
\newlength{\mgd}
\title{On the maximum number of minimal connected dominating sets in convex bipartite graphs \footnote{%
This work is supported by the French National Research Agency  
by the ANR project GraphEn (ANR-15-CE40-0009).}}

\titlerunning{On the maximum number of minimal connected dominating sets in convex bipartite graphs}

\author{Mohamed Yosri Sayadi}{Universit\'e de Lorraine, LGIPM, F-57000 Metz, France}{yosri.sayadi@univ-lorraine.fr}{}{}
\authorrunning{M. Y. Sayadi }

\Copyright{Mohamed Yosri Sayadi}

\subjclass{G.2.2, F.2.2}

\keywords{Minimal connected dominating set, exact algorithms, enumeration, graphs}

\ArticleNo{1}
\nolinenumbers 
\hideLIPIcs  

\begin{document}
\maketitle

\begin{abstract}
The enumeration of minimal connected dominating sets is known to be notoriously hard for general graphs. Currently, it is only known that the sets can be enumerated slightly faster than $\mathcal{O}^{*}(2^n)$ and the algorithm is highly nontrivial. 
Moreover, it seems that it is hard to use bipartiteness as a structural aide
when constructing enumeration algorithms. Hence, to the best of our knowledge, there is no known input-sensitive algorithm for enumerating minimal dominating sets, or one of their related sets, in bipartite graphs better than that of general graphs.
In this paper, we provide the first input-sensitive enumeration algorithm for some non trivial subclass of bipartite graphs, namely the convex graphs. 
We present an algorithm to enumerate all minimal connected dominating sets of convex bipartite graphs in
time $\mathcal{O}(1.7254^{n})$ where $n$ is the number of vertices of the input graph.
Our algorithm implies a corresponding upper bound for the number of minimal connected dominating sets for this graph class.
We complement the result by providing a convex bipartite graph, which have at least $3^{(n-2)/3}$ \mcdss{}.
\end{abstract}

\section{Introduction} \label{sec:intro}
Listing, generating or enumerating objects of specified type and properties has important
applications in various domains of computer science, such as data mining, machine learning, and artificial
intelligence, as well as in other sciences, especially biology.
In particular, enumeration algorithms whose running time is measured in the size of the input
have gained increasing interest recently \cite{Golovach2018, GolovachKLS17, GolovachKLS19}.

In fact, several classical examples exist in this direction, of which one of the most famous is perhaps
that of Moon and Moser \cite{Moon1965} who showed that the maximum number of maximal independent sets
in a graph on $n$ vertices is $\Theta (3^{n/3})$.
Recently Lokshtanov et al. \cite{Lokshtanovmcds} studied the maximum number of minimal connected dominating
sets in an arbitrary $n$-vertex graph and they showed that the maximum
number of minimal connected dominating sets is $\mathcal{O} (     2^{  (1-\epsilon )n   }    )$ where $\epsilon > 10^{-50}$.
Contrary to the number of maximal independent sets where the upper and lower bounds are tight, 
the best known lower bound for the number of minimal connected dominating sets is $3^{(n-2)/3}$ \cite{ch+mcds},
meaning there is a huge gap up to now between the lower and upper bounds.
It is worth noting that computing a minimum connected dominating set is one of the classical NP-hard problems as already
mentioned in the monograph of Garey and Johnson \cite{Garey1979}. Furthermore, it also is NP-hard for bipartite graphs \cite{bipartite}
and chordal bipartite graphs \cite{Haiko}.
The best known running time of an algorithm solving this problem for general graphs is $O(1.8619^{n} )$ \cite{Abu-Khzam:2011:EAC:1994013.1994228}. 
However, the minimum connected dominating set problem is tractable for convex bipartite graphs and it can be solved in
$\mathcal{O} ( n^{3})$ time \cite{DamaschkeMK90}. 
Furthermore, the problems of finding a minimum dominating set and a minimum
independent dominating set in an $n$-vertex convex bipartite graph are solvable in time
$\mathcal{O} ( n^{2})$ \cite{jensen}.

Despite the structural and algorithmic properties of bipartite graphs, there is no enumeration algorithm for the \mcdss{} in a
bipartite graph better than that for general graphs. The situation is similar when considering the enumeration of
minimal dominating sets and maximal irredundant sets in bipartite graphs.
In this paper we study the enumeration and maximum number of minimal connected dominating
sets in convex bipartite graphs, and we prove that the number of minimal connected dominating sets in a convex bipartite graph
is $ \mathcal{O}(1.7254^{n})$ and that those sets can be enumerated in time
$ \mathcal{O}(1.7254^{n})$. 

The studied graph classes for the enumeration of minimal connected dominating sets are summarized in
the following table, where $n$ is the number of vertices and $m$ is the number of edges of an input graph
belonging to the given class.

 \begin{table}[h!]
  \begin{center}
    \caption{Lower and upper bounds on the maximum number of minimal connected dominating sets.}
    \label{tab:table1}
    \begin{tabular}{l l l l l} 
      \textbf{Graph class} & \textbf{Lower bound} & \textbf{Ref.}& \textbf{Upper bound} & \textbf{Ref.}\\
    \hline
      general 			& $3^{(n-2)/3}$ 	& \cite{ch+mcds}& $\mathcal{O} (     2^{  (1-\epsilon )n   }    )$ 		& \cite{Lokshtanovmcds}\\
      chordal 			& $3^{(n-2)/3}$ 	& \cite{ch+mcds}& $\mathcal{O} (  1.7159^{n}) $ 					& \cite{ch+mcds} \\
      split 				& $1.3195^{n}$ 	& \cite{CouturierHHK13}& $1.3674^{n}$ 								& \cite{sp+mcds}\\
      co-bipartite 		& $1.3195^{n}$ 	& \cite{CouturierHHK13}& $n^{2}+2\times 1.3674^{n}$ 					& \cite{sp+mcds}\\
      interval 			& $3^{(n-2)/3}$ 	&\cite{ch+mcds}& $3^{(n-2)/3}$ 								& \cite{ch+mcds}\\
      AT-free 			& $3^{(n-2)/3}$ 	&\cite{ch+mcds}& $\mathcal{O}^{*} ( 3^{(n-2)/3} ) $ 				& \cite{ch+mcds}\\
      strongly chordal	& $3^{(n-2)/3}$ 	&\cite{ch+mcds}& $ 3^{n/3} $ 								& \cite{ch+mcds}\\
      distance-hereditary	& $3^{(n-2)/3}$ 	&\cite{ch+mcds}& $ 3^{n/3}\times n $ 							& \cite{ch+mcds}\\
      cograph			& $m$		 	& \cite{ch+mcds}& $ m$										&\cite{ch+mcds}\\
      convex bipartite		& $3^{(n-2)/3}$     & [this paper] & $\mathcal{O} ( 1.7254^{n} )$ 								& [this paper]\\

    \end{tabular}
  \end{center}
\end{table}

\section{Preliminaries}\label{sec:defs}

We consider finite undirected convex bipartite graphs $G=(U,W,E)$ without loops or multiple edges. 
Let $ V = U \cup W $ be the vertex set of $G$, where $U$ and $W$ define the bipartition of vertices.
We also let $n = |V (G)|$ and $m = |E(G)|$ denote the number of vertices and edges, respectively, of the input graph $G$. 
A bipartite graph $G = (U, W, E)$ is convex if there exists an ordering of the vertices of $W$
such that for each $u \in U$, the neighbors of $u$ are consecutive in $W$.
 For convenience, we consider that $U = \{1, 2, \ldots, |U |\}$ and $W = \{w_1, w_2, \ldots, w_{|W |}\}$, and that the vertices in $W$
are given according to the ordering mentioned above. 
We say that a vertex $w_{i} \in W$ is smaller (larger) than a vertex $w_{j} \in W$ if the integer $i$ is
smaller (larger) than the integer $j$.
By the definition of convex bipartite graphs, the neighbors of a vertex $u \in U$ can be represented
as an interval $I_{u} = [l(I_u), r(I_u)]$ , called the neighbor interval of $u$,
where $l(I_u)$  and $r(I_u)$ are the smallest and largest vertices, respectively, in the interval
of vertices of $W$ adjacent to $u$. Further, we call $l(I_u)$ and $r(I_u)$ the left endpoint and right endpoint of the interval $I_u$, respectively.
Then, the neighbors of vertices of $U$ can be represented by a set $I (U )$
of intervals. We call $I (U )$ the interval representation of neighbors of vertices in $U$.
A star graph is the complete bipartite graph $K_{1,n-1}$: a tree with one internal node and $n-1$ leaves.

Let $v$ be a vertex of $V$ and let $D$ be a subset of $V$. The open neighborhood $N(v)$ 
of the vertex $v$ consists of the set of vertices adjacent to $v$, that is, 
$N(v)=\{w \in V |vw \in E\}$, and the closed neighborhood of $v$ is $N [v] = N (v) \cup \{v\}$. 
The open neighborhood $N(D)$ is defined to be $\cup_{v\in D} N(v)$, and the closed neighborhood of $D$
is $N [D] = N(D) \cup D$. We denote by $deg(v)=|N(v)| $ the degree of $v$.
The set $D$ dominates vertex $v$ if either $v \in D$ or $N (v) \cap D \ne \emptyset$.
If $D$ dominates all vertices in a subset $S$ of $V$, then we say that $D$ dominates $S$. The
set $D$ is called a dominating set of $G$ if and only if $D$ dominates $V$.
A vertex $u\in V$ is a private neighbor of the vertex $v$ (with respect to $D$) if $u$ is dominated by $v$
but is not dominated by $D \setminus \{v\}$.
We denote by $G[D]$ the subgraph of $G$ induced by $D$. 
A set $D \subseteq V(G)$ is connected if $G[D]$ is a connected
graph. 
A set of vertices $D$ is a connected dominating set if $D$ is connected and $D$ dominates $V (G)$.
A (connected) dominating set is minimal if no proper subset of it is a (connected) dominating set.
Further, a connected dominating set $D$ is minimal if and only if for any $v \in D$, $v$ has a private neighbor or $D \setminus \{v\}$ is
disconnected, i.e., $v$ is a cut vertex of $G[D]$.
Notice that every cut vertex of $G$ belongs to every \mcds{} of $G$.

It is possible to recognize convex bipartite graphs in linear time\cite{Brandstadt:1999}.
We use the $O^*$ notation that hides polynomials, i.e., we write $f(n)=O^*(g(n))$ if $f(n)=O(p(n) \cdot g(n))$ where $p$ is a polynomial in $n$.

\section{Structural Properties}

In this section, we will provide some useful properties which any \mcds{} $D$ of a convex bipartite graph $G=(U,W,E)$ satisfies.
Recall that the vertices $W$ of the graph $G=(U,W,E)$ are given with an ordering $w_1, w_2, \ldots, w_{|W |}$ satisfying
that for all $u \in U$, $N(u) \subset W$ is an interval $I_u=[l(u),r(u)]$.

\begin{observation}\label{lemStar}
A set $D\subseteq V$ is a \mcds{} of a $K_{1,n-1}$ if and only if $D$ is a singleton consisting of the internal node of $K_{1,n-1}$.
Hence, if $G$ is a star graph then 
the problem can be solved in linear time.
\end{observation}

From now on, we consider a convex bipartite graph $G=(U,W,E)$ 
satisfying $|W| \ge 2$ and $|U| \ge 2$ which implies $|D| \ge 2$
for each \mcds{} of $G$. 

\begin{observation}\label{lemN}
 Each vertex $v \in V$ has a neighbor $x\in D$.
\end{observation}
\begin{proof}

 If $v \not\in D$, then there exists an $x \in D \cap N(v)$ because $D$ is a dominating set.
 Otherwise $v \in D$, and then since $D$ is a connected set and 
 $|D| \ge 2$, there exists an $x \in D \cap N(v)$.
\end{proof}

\begin{lemma}\label{lemIncD}
 If $i,j \in U$ such that $I_i \subseteq I_j$ then 
 $|\{i,j\} \cap D| \le 1$.
\end{lemma}
\begin{proof}
  
Let $i,j\in U: I_i \subseteq I_j$ and suppose, for the sake of contradiction, that both $ i,j \in D$.
The vertex $i$ cannot have a private neighbor with respect to $D$ in $W$ because all its neighbors are dominated by $j$.
Furthermore, $i$ cannot be a private for itself because $D$ is connected and $|D| \ge 2$.
Therefore $i$ cannot have a private neighbor.
On the other hand, to show that $i$ is not a cut vertex of $G[D]$ let $x,y \in V \setminus \{i\}$. For every $x-y$ path in $G$ which passes through $i$, there exists
a $x-y$ path passing through $j$
without passing through $i$. Therefore $i$ cannot be a cut vertex of $G[D]$.
Hence $D$ is not a \mcds{}; a contradiction.
\end{proof}

This immediately implies 

\begin{lemma}\label{lemIncR}
 If $r(I_i) = r(I_j)$ or $ l(I_i) = l(I_j)$, $ i,j \in U $, then $|\{i,j\} \cap D| \le 1$. 
\end{lemma}

The following two lemmata are crucial for our branching algorithm. 

\begin{lemma}\label{lemCov}
For every $i \le |W|-1$, there exists a vertex $u \in D \cap U$ such that $u \in N(w_i) \cap N(w_{i+1})$.
\end{lemma}
\begin{proof}
Suppose, for the sake of contradiction, there exists an $i \le |W|-1$ such that $D \cap N(w_i) \cap N(w_{i+1}) = \emptyset $.
Let us consider the induced subgraph  $H=G[W\cup (U\cap D)]$.  We claim that $H'=(U',W',E')$ and $H\dprime=(U\dprime,W\dprime,E\dprime)$
such that $U'= \{u'\in U :r(I_{u'}) \le w_i \} $,
$W'=\{w_{i'} \in W : w_{i'} \le w_{i}\} $,
$U\dprime= \{u\dprime \in U:l(I_{u\dprime}) \ge w_{i+1} \} $ and
$W\dprime=\{w_{i\dprime} \in W : w_{i\dprime} \ge w_{i}\} $, are two connected components of  $H=(U,W,E)$.
To see this note that there is no $w_ju \in E$ such that $w_j \in W', u \in U\dprime$ or $w_j \in W\dprime, u \in U'$, otherwise $ u \in N(w_i) \cap N(w_{i+1})$.
Hence the induced subgraphs $H'$ and $H\dprime$ are indeed two disjoint components of $H$.
By observation~\ref{lemN}, $ N(w_i) \cap  D \ne \emptyset $ and $ N(w_{i+1}) \cap D \ne \emptyset $, thus $D \cap U' \ne \emptyset$ and
$ D \cap U\dprime \ne \emptyset$. Hence $G[D]$ is disconnected, a contradiction.
\end{proof}

Similarly, we may obtain the following lemma.

\begin{lemma}\label{lemCovD}
For any two consecutive $w_i, w_j \in W \cap D $ in $G[D]$ (not necessarily consecutive in $W$), there exists a vertex $u \in D \cap U$
such that $u \in N(w_i) \cap N(w_{j})$.
\end{lemma}

\begin{lemma}\label{obsForbid}
 
Let $i,j \in U$ such that $I_i \cap I_j \neq \emptyset$.
If there is a $k \in U$ such that $I_k \subset (I_i \cup I_j)$ and $i,j,k \in D \cap U$, then $I_i \cap I_j \cap D = \emptyset$.
 
\end{lemma}

\begin{proof}
 
 Let $i,j,k \in D\cap U$, $I_i \cap I_j \ne \emptyset$ and $I_k \subset \{I_i \cup I_j\}$. Suppose, for the sake of
 contradiction, that there is a $ w \in D \cap I_i \cap I_j$.
 It is clear that $k$ is a cut vertex of $G[D]$. Therefore there are $w_x,w_x' \in D \cap I_k$ and $w_x,w_x'$ are in two different
 components of $G[D \setminus k]$. Furthermore $ w_x \in I_i \cap I_k \setminus I_j$ and $ w_x' \in I_j \cap I_k \setminus I_i$.
 However there is a $w_x-w_x'$ path passing through $i,w,j$ without passing through $k$, a contradiction.
\end{proof}

\begin{lemma}\label{obsOne}
 
Let   $i,j \in D \cap U $ such that $I_i \cap I_j \ne \emptyset$. 
Then $ | \{ k \in D : I_k \subset (I_i \cup I_j) \}| \le 1$.

\end{lemma}

\begin{proof}
 Let $i,j \in D$, $I_i \cap I_j \ne \emptyset$ and suppose, for the sake of contradiction, that there are $k,l \in D$ satisfying $I_k \subset \{I_i \cup I_j\}$
 and $I_l \subset \{I_i \cup I_j\}$.
 Because of lemmata \ref{lemIncR} and \ref{lemIncD} and w.l.o.g, then $ l(I_i) <l(I_k) < l(I_l) < l(I_j)$ and $ r(I_i) < r(I_k) < r(I_l) < r(I_j)$. 
It is clear that $k,l$ are cut vertices of $G[D]$.
As $k$ is a cut vertex of $G[D]$, there are $ w_x,w_{x'} \in D \cap I_k$  and $w_x,w_x'$ in two different components of $G[D \setminus k]$.
Furthermore, at least one of them does not belong to $I_l$. W.l.o.g, let $w_x < w_{x'}$ and $w_x \notin I_l$.
Since $l$ is a cut vertex of $G[D]$, there are $w_y,w_{y'} \in D \cap I_l$  and $w_y,w_{y'}$ are in two different components of $G[D \setminus l]$.
It is clear that at least one of $w_y,w_{y'}$ is not a neighbor of $k$.
W.l.o.g, let $w_{y'} > w_{y}$ and $w_{y'} >  r(I_k) $.
Furthermore, $w_{y} < l(I_j)$ and $w_{y}> w_x$ otherwise $w_x,w_{x'}$ will be adjacent to $l$ and $w_y,w_{y'}$ will be adjacent to $j$.
It is clear that $w_{x'} > r(I_i)$.
Hence, there is a $w_{y}-w_{y'}$ path passing through $k,w_{x'},j$ without passing through $l$, a contradiction.
\end{proof}

Now we are ready to present our algorithm.

\section{The enumeration algorithm}

The basic idea of the enumeration algorithm is to choose the vertices of a \mcds{} $D$ by using reduction and branching rules.
Furthermore the algorithm is partitioned into stages. During the preprocessing (stage 1) a collection of initializing recursive calls is done. In stage 2 we choose the vertices of $D \cap U$. 
Therefore when we select a vertex $u\in U$, we add it immediately
to the \mcds{} $D$. However when we discard $u$ from $D$, we move it to $T$ in order to dominate it in the next stages of the algorithm by a (still to be selected) vertex of $W$.
Furthermore when we fix the vertices of $D \cap U$, we mark some vertices of $W$ as forbidden vertices $F$, which means that those vertices of $W$ cannot be selected (they are excluded from D). 
In stage 3 we mainly remove all those partial solutions of stage 2
that cannot be extended to a \mcds{} of $G$. Finally in stage 4
the remaining partial solutions are completed into \mcds{}, if possible. 

\medskip
\underline{{\bf \textit{Stage~1. Preprocessing.}}}

\medskip
\noindent
We consider the following procedure $\textsc{EnumLevel1}(U,W)$, where we initialize the call of $\textsc{EnumLevel2}(u,U,D,T,F)$.

\medskip
\noindent
{\bf Step~1.} for each $ u \in N(w_1)$, call  $\textsc{EnumLevel2}(u,U\setminus N(w_1),\{u\},N(w_1)\setminus \{u\},\emptyset)$.

\medskip
\noindent
By observation \ref{lemN}, $|N(w_1) \cap D| \ge 1$.
Since all the neighbors of $w_1$ have the same left endpoint, $ |N(w_1) \cap D| \le 1$ by lemma \ref{lemIncR}.
Therefore $ |N(w_1) \cap D| = 1$. Hence we add exactly one vertex $u \in N(w_1)$ to $D$ and discard the other
vertices $N(w_1) \setminus \{u\}$ from the solution, i.e.
we are moving them to $T$ in order to dominate
them by vertices of $W$ in the next stages of the algorithm.
Finally, we initialize $F=\emptyset$ and call $\textsc{EnumLevel2}(u,U,D,T,F)$.

\medskip
\underline{{\bf \textit{Stage~2.}}} 

\medskip
\noindent
We consider the following recursive procedure $\textsc{EnumLevel2}(u,U,D,T,F)$, where $u$ is the vertex already selected
in $D\cap U$ with the largest right endpoint and $U,D,T,F$ were mentioned above. Let $r(I_u)$ be $r$.

\medskip
\noindent
{\bf Step~1.} If there is a vertex $i\in U$ such that $r(I_i) \le r$, then call $\textsc{EnumLevel2}(u,U\setminus\{i\},D,T\cup \{i\},F)$.

\medskip
\noindent
{\bf Step~2.} If $ r = w_{|W|}$, then call $\textsc{EnumLevel3}(D,T,F)$.

\medskip
\noindent
{\bf Step~3.} If $ N(r ) = \emptyset$, then stop.

\medskip
\noindent
{\bf Step~4.} If $ deg(r) = 1 $, then let $ \{j\} = N(r)$ and call
$\textsc{EnumLevel2}(j,U\setminus\{j\},D\cup \{j\},T,F)$.

\medskip
\noindent
{\bf Step~5.} If $ deg(r) = 2$, then let $ \{j,k\} = N(r)$ and branch as follows:

 \hspace*{\mg}If $ I_j \subseteq I_k $ or $ I_k \subseteq I_j $, then branch:
\begin{itemize} [leftmargin=\mgd]
\item[(i)] call $\textsc{EnumLevel2}(j, U\setminus\{j,k\}, D\cup\{j\},T \cup\{k\},F)$,
\item[(ii)] call  $\textsc{EnumLevel2}(k,U\setminus\{j,k\},D\cup\{k\},T\cup\{j\},F)$.
\end{itemize}

\hspace*{\mg}Else, let $r(j)>r(k)$ and branch:
\begin{itemize} [leftmargin=\mgd]
\item[(i)] call $\textsc{EnumLevel2}(j, U\setminus\{j,k\}, D\cup\{j\},T \cup\{k\},F)$,
\item[(ii)] call  $\textsc{EnumLevel2}(k,U\setminus\{j,k\},D\cup\{k\},T \cup\{j\},F)$.
\item[(iii)] call  $\textsc{EnumLevel2}(j,U\setminus\{j,k\},D\cup\{j,k\},T,F \cup ( I_j \cap I_u) )$.
\end{itemize}

\medskip
\noindent
{\bf Step~6.} If $ deg(r) \ge 3$, then let $j$ be the neighbor of $r$ with the largest right endpoint and branch as follows:
\begin{itemize}[leftmargin=\mgd]
\item[(i)] call $\textsc{EnumLevel2}(j, U\setminus N(r), D\cup\{j\},T \cup N(r) \setminus \{j\},F)$,
\item[(ii)] for each $x \in N(r)$ with $I_x \not \subseteq I_j~$, call\\ 
$\textsc{EnumLevel2}(j,U\setminus N(r) ,D\cup\{x,j\},T \cup N(r) \setminus \{x,j\},F \cup \{I_j \cap I_u\})$.
\item[(iii)] call  $\textsc{EnumLevel2}(u,U\setminus\{j\},D,T \cup\{j\},F)$.
\end{itemize}

\noindent
Let $j$ be the neighbor of $r$ with the largest right endpoint. Either $j \in D$ or $j \notin D$.\\
If $j \notin D$, then in case (iii) we add $j$ to $T$ in order to be dominated in the next stages by vertices of $W$.
Therefore we call $\textsc{EnumLevel2}(u,U\setminus\{j\},D,T \cup\{j\},F)$.\\
Suppose now that $j \in D$. Because $j$ is the neighbor of $r$ with the largest right endpoint then
$I_u \cap I_j \neq \emptyset $ and for any $x \in D \cap N(r) \setminus \{u,j\} $ we have $I_x \subset (I_u \cup I_j) $.
Therefore by lemma \ref{obsOne}, $| N(r) \cap D \setminus \{u,j\} | \le 1$.
If $ N(r) \cap D \setminus \{u,j\} = \emptyset$, then in case (i) we call $\textsc{EnumLevel2}(j, U\setminus N(r), D\cup\{j\},T \cup N(r) \setminus \{j\},F)$.
Suppose now that $ N(r) \cap D \setminus \{u,j\} = \{x\}$.
In this case, we forbid by lemma \ref{obsForbid} $ ( I_j \cap I_u ) $.
Notice that by lemma \ref{lemIncD} $ I_x \not \subseteq I_j $.
Hence we branch for each $x \in N(r)$ satisfying $I_x \not \subseteq I_j~$ and we call\\
$\textsc{EnumLevel2}(j,U\setminus N(r) ,D\cup\{x,j\},T \cup N(r) \setminus \{x,j\},F \cup \{I_j \cap I_u\})$.
If there is an $x \in U$ such that $I_x \subseteq I_j$, then $x$ will be treated by step 1 in the recursive call.

\medskip
\underline{{\bf \textit{Stage~3.}}}  

\medskip
\noindent
In this stage, the procedure $\textsc{EnumLevel3}(D,T,F)$ deletes the bad partial solutions, i.e. those that can definitely not be extended
into a \mcds{} of $G$, generated in the previous stage and it preprocesses the remaining partial solutions for the next stage. 
Let $J(T,D)$ be a set of intervals, where for every interval we have to select at least
one corresponding vertex either to dominate a vertex in $T$ or to connect vertices in $D$.
Hence we initialize $J(T,D)$ by the interval representations of neighbors of vertices in $T$ and consider the induced
graph $ G[W\cup (U\cap D)]$.
Let $I'$ be the interval representations of neighbors of vertices in $D\cap U$.

\medskip
\noindent
{\bf Step~1.} If there exist $u,v \in D\cap U$ such that $I'_u \subseteq I'_v$, then stop.

\medskip
\noindent
{\bf Step~2.} While ( $I' \ne \emptyset$ )\\
\hspace*{\mg} Begin\\
\hspace*{\mg}
Let $I'_i$ be the interval with the smallest right endpoint in $I'$ and $r(I'_i)$ be $r$.\\
\hspace*{\mg}
If $I'=\{I'_i\}$, then call $\textsc{EnumLevel4}(D,J,F)$.\\
\hspace*{\mg}
Else if $ deg(r) > 3 $, then stop.\\
\hspace*{\mg}
Else if $ deg(r)=2 $, then let $ N(r)=\{i,j\}$ :\\
\hspace*{\mgd}$J(T,D) \leftarrow J(T,D) \cup \{I'_i \cap I'_j\} $ ;\\
\hspace*{\mgd}$I' \leftarrow I' \setminus I'_i$.\\
\hspace*{\mg}
Else if $ deg(r)=3 $, then let $ N(r)=\{i,j,k\}$, $ I'_j$ be the interval with the largest right
\hspace*{\mg} endpoint and
$I'_k$ be the other interval such as $I'_k \subset (I'_i \cup I'_j)$ :\\
\hspace*{\mgd}$J(T,D)  \leftarrow J(T,D) \cup ( I'_i \cap I'_k  \setminus I'_j) \cup ( I'_j \cap I'_k  \setminus I'_i) $ ;\\
\hspace*{\mgd}$I' \leftarrow I' \setminus (I'_i \cup I'_k) $.\\
\hspace*{\mg}
End.\\

\noindent
 First, every vertex in $T$ should be dominated by at least one of its neighbors. Therefore we initialized $J(T,D)$ by the interval
 representations of neighbors of vertices in $T$.\\
 Now if $ deg(r)> 3 $ in the induced graph $ G[W\cup (U\cap D)]$, then $D$ cannot be a \mcds{} by lemma \ref{obsOne}. Therefore we stop.\\
 So let $ deg(r)=2 $:\\
 Suppose, for the sake of contradiction, that $I'_i \cap I'_j \cap D = \emptyset$.
 Then by observation \ref{lemN} there is a $w_i \in I'_i \cap D$ and there is a $ w_j \in I'_j \cap D$. W.l.o.g, let $w_i,w_j$ be consecutive in $G[D]$.
 As $deg(r)=2$, there cannot be any $u \in U\cap D$ such that $u \in N(w_i) \cap N(w_j)$, a contradiction by lemma \ref{lemCovD}.
 Therefore $I'_i \cap I'_j \cap D \ne \emptyset$.
 Hence $J(T,D) \leftarrow J(T,D) \cup \{I'_i \cap I'_j\}$.\\
 Now if $ deg(r)=3 $, then the vertices of $I'_i \cap I'_j$ are forbidden and $k$ is a cut vertex of $G[D]$ which connects two vertices in $W \cap D$.
 If both of these two vertices had belonged to $( I'_i \cap I'_k  \setminus I'_j)$ or $( I'_j \cap I'_k  \setminus I'_i) $, then $k$ would not have been a cut vertex.
 Therefore one belongs to $( I'_i \cap I'_k  \setminus I'_j)$ and the another one belongs to $( I'_j \cap I'_k  \setminus I'_i) $.
 Hence, we added $( I'_i \cap I'_k  \setminus I'_j)$ and $( I'_j \cap I'_k  \setminus I'_i) $ to $J(T,D)$.

\medskip
\underline{{\bf \textit{Stage~4.}}}  

\medskip
\noindent
We consider the following recursive procedure called $\textsc{EnumLevel4}(W,J,D)$.
We try, in this level, to select at least one vertex of each interval in $J$ in order to dominate the vertices of $T$ and to connect $D$
by some vertices of $W$.\\
\noindent
Let $J_i \in J$ be the interval with the smallest right endpoint. If we have more than one candidate interval, then the shortest interval amongst them
will be chosen.

\medskip
\noindent
{\bf Step~1.} If $J=\emptyset$, then check whether $D$ is a \mcds{} of $G$ and output it if it holds; then stop.

\medskip
\noindent
{\bf Step~2.} If all the vertices of $J_i$ are forbidden, then stop.

\medskip
\noindent
{\bf Step~3.} For each non forbidden $w \in J_i$,
call $\textsc{EnumLevel4}(W \setminus J_i, J\setminus \{J_k:w\in J_k \},D\cup\{w\})$.\\

\noindent
By the construction of $J$, we must select at least one vertex of each interval in $J$. Therefore $ |D \cap J_i| \ge 1$.
Hence in Step~2, we stop when all the vertices of $J_i$ are forbidden.
Let us prove now that $ |D \cap J_i| \le 1$.
Suppose, for the sake of contradiction, that $ |D \cap J_i| \ge 2$. Therefore there are $ w_i,w_j \in J_i \cap D$.
W.l.o.g, we suppose that $ w_i < w_j $.
In this case if $ w_i \in J_k$ such that $J_k \in J$, then $w_j \in J_k$.
Therefore, $D$ is not minimal because we can delete $w_i$ and $D$ is still a connected dominating set, a contradiction.
Hence $ |D \cap J_i| \le 1$ and because $ |D \cap J_i| \ge 1$ then $ |D \cap J_i| = 1$.

\section{Running time and upper bound}

We establish an upper bound on the number of minimal connected dominating sets in convex bipartite graphs, via the branching
algorithm described in the previous section and its running-time analysis.

For the analysis of the running time and the number of minimal connected dominating sets that are produced by such an algorithm, we use a technique based on
solving recurrences for branching steps and branching rules respectively. 
We refer to the book by Fomin and Kratsch \cite{FominK10} for a detailed introduction. 
To analyze such a branching algorithm solving an enumeration problem, one assigns
to each \emph{instance} $I$ of the recursive algorithm a \emph{measure} $\mu(I)$ that one may 
consider as the size of the instance $I$. 
If the algorithm branches on an instance $I$ into $t$ new instances, such that the  
measure decreases by $c_1, c_2, \ldots, c_t$ for each new instance, respectively, we say that
$(c_1, c_2, \ldots, c_t)$ is the \emph{branching vector} of this step.
We find the unique positive real root $\alpha$, called a \emph{branching number}, of
the \emph{characteristic polynomial} $p(x)=x^c-x^{c-c_1}- \ldots - x^{c-c_t}$ for $c=\max\{c_1,\ldots,c_t\}$. 
Then standard analysis (see~\cite{FominK10}), shows that if $\mu(I)\leq n$ for all instances $I$, the number of leaves of the search tree
produced by an execution of the algorithm is $O^*(\alpha^n)$, where $\alpha$ is the maximum value of the branching numbers over all branching vectors that occur in the algorithm. 
This approach allows us to  achieve running times of the form $O^*(\alpha^n)$ for some real $\alpha \ge 1$. As the number of \mcdss{} produced by
an algorithm is upper bounded by the number of leaves of the search tree, we also obtain the upper bound for the number of \mcdss{} of the same
form $O^*(\alpha^n)$. If $\alpha$ has been obtained by rounding up then one may replace $O^*(\alpha^n)$ by $O(\alpha^n)$; see \cite{FominK10}.

To analyze the running time of the algorithm, we compute the branching vectors for all branching steps of
the procedures $\textsc{EnumLevel4}(W,J,D)$ and $\textsc{EnumLevel2}(u,U,D,T,F)$.
Notice that $\textsc{EnumLevel1}$ and $\textsc{EnumLevel3}$ runs in polynomial time as no branching is needed.
We set the measure of an instance to $|U|+|W\setminus F|$. Hence by moving a vertex $u \in U$ to $T$ or by forbidding a vertex $w\in W$, the measure of the instance
decreases by $1$.

Let us start with $\textsc{EnumLevel2}(u,U,D,T,F)$.
Notice that in Steps $1 - 4$ we reduce an input without branching (reduction rules).
Hence, to analyze the time, we only have to analyze Steps 5 and 6 (branching rules).

\smallskip
\noindent
{\bf Step~5.} The first branching vector in Step~5 is $(2,2)$. In the second branching rule, in the worst case we forbid only one vertex in $W$.
Thus the branching vector is $(2,2,3)$. 
Hence the maximum value of the branching numbers is achieved for $(2,2,3)$ and thus $\alpha_5<1.6181$.

\smallskip
\noindent
{\bf Step~6.} It is straightforward to see, that the maximum value of the branching number is
achieved if $|I_u \cap I_j|=1$ and if there is no $ I_x \subseteq I_j$ such that $x,j \in N(r) $. Thus we branch for all the neighbors of $r$ and we forbid
in each case the only vertex $I_u \cap I_j$.
The corresponding branching vector is $(t, \underbrace{t+1,\ldots,t+1}_{t-1},1 )$ where $t= deg(r) \ge 3$, and the maximum value of the branching
number $\alpha_6< 1.7254$ is achieved for $t=3$.

For $\textsc{EnumLevel4}(W,J,D)$, we need to analyze Step 3 only because Steps 1 and 2 are reduction rules.
The corresponding branching vector is $( \underbrace{t,\ldots,t}_{t} )$ where $t=|J_i \setminus F|$, and the maximum value of the branching
number $\alpha_3< 1.4423$ is achieved for $t=3$.

\smallskip
The largest branching number is (majorized
by) $1.7254$. Thus we may conclude with the following theorem. 

\begin{theorem}\label{thm}
A convex bipartite graph has at most $\mathcal{O}(1.7254^{n})$ \mcdss{},
and these can be enumerated in time $\mathcal{O}(1.7254^n)$.
\end{theorem}

\section{Lower bound}

To obtain a lower bound for the maximum number of \mcdss{} in a convex bipartite graph, we use a slight modification of the lower bound obtained in \cite{ch+mcds}.

\begin{proposition}

There are convex bipartite graphs with at least $3^{(n-2)/3} $ minimal connected dominating sets.

\end{proposition}

\begin{proof}
To obtain the bound for convex bipartite graphs, consider the graph $G$ constructed as follows for a positive odd integer $k$.
\begin{itemize}
 \item For $i \in \{1,\dots,k\}$, construct a triple of independent vertices $ T_i = \{x_i, y_i, z_i\}$.
 \item For $i \in \{2,\dots,k\}$, join each vertex of $T_{i-1}$ with every vertex of $T_i$ by an edge.
 \item Construct two vertices $u$ and $v$ and edges $ux_1$, $uy_1$, $uz_1$ and $vx_k$, $vy_k$, $vz_k$.

\end{itemize}

Clearly, $G$ has $n = 3k + 2$ vertices. Notice that $D \subseteq V (G)$ is a minimal connected dominating set of $G$ if and only if $u, v \notin D$ and
$|D \cap T_i| = 1$ for $i \in {1,\dots,k}$. Therefore, $G$ has $3^{k} = 3^{(n-2)/3}$ minimal connected dominating sets.
It remains to observe that $G$ is convex bipartite from its following model. 
\end{proof}

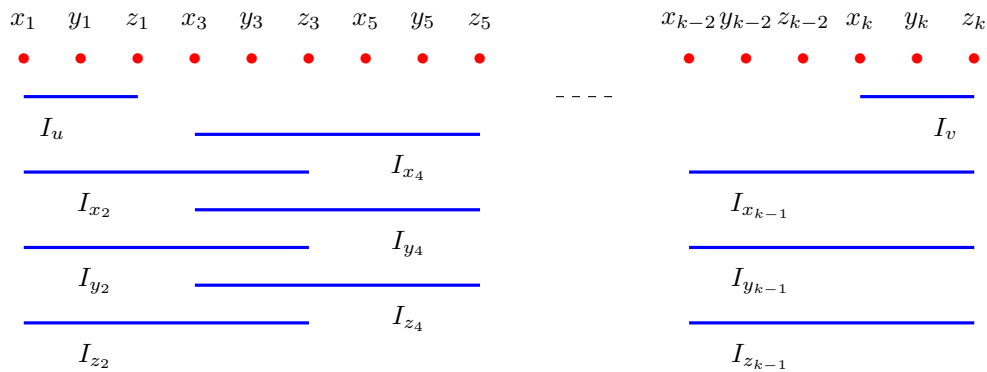
\begin{tikzpicture}
    \node (x_1) at (0,0) {$x_1$};
    \draw (0,-0.5) node[red]{$\bullet$};
    
    \node (y_1) at (0.75,0) {$y_1$};
    \draw (0.75,-0.5) node[red]{$\bullet$};

    \node (z_1) at (1.5,0) {$z_1$};
     \draw (1.5,-0.5) node[red]{$\bullet$};

    \node (x_3) at (2.25,0) {$x_3$};
    \draw (2.25,-0.5) node[red]{$\bullet$};

    \node (y_3) at (3,0) {$y_3$};
    \draw (3,-0.5) node[red]{$\bullet$};

    \node (z_3) at (3.75,0) {$z_3$} ;
    \draw (3.75,-0.5) node[red]{$\bullet$};

    \node (x_5) at (4.5,0) {$x_5$};
        \draw (4.5,-0.5) node[red]{$\bullet$};

    \node (y_5) at (5.25,0) {$y_5$};
        \draw (5.25,-0.5) node[red]{$\bullet$};

    \node (z_5) at (6,0) {$z_5$} ;
        \draw (6,-0.5) node[red]{$\bullet$};

            \node (x_k2) at (8.75,0) {$x_{k-2}$};
        \draw (8.75,-0.5) node[red]{$\bullet$};

    \node (y_k2) at (9.5,0) {$y_{k-2}$};
        \draw (9.5,-0.5) node[red]{$\bullet$};

    \node (z_k2) at (10.25,0) {$z_{k-2}$} ;
        \draw (10.25,-0.5) node[red]{$\bullet$};

    \node (x_k) at (11,0) {$x_k$};
        \draw (11,-0.5) node[red]{$\bullet$};

    \node (y_k) at (11.75,0) {$y_k$};
        \draw (11.75,-0.5) node[red]{$\bullet$};

    \node (z_k) at (12.5,0) {$z_k$} ;
        \draw (12.5,-0.5) node[red]{$\bullet$};

    \draw [draw=blue, very thick] (0,-1) -- (1.5,-1) node [ near start, label=below :{$I_u$}] {};
    
   \draw [draw=blue, very thick] (0,-2) -- (3.75,-2) node [near start, label=below :{$I_{x_2}$}] {};
   
   \draw [draw=blue, very thick] (0,-3) -- (3.75,-3) node [near start, label=below :{$I_{y_2}$}] {};
   
   \draw [draw=blue, very thick] (0,-4) -- (3.75,-4) node [near start, label=below :{$I_{z_2}$}] {};

   \draw [draw=blue, very thick] (2.25,-1.5) -- (6,-1.5) node [near end, label=below :{$I_{x_4}$}] {};
   
   \draw [draw=blue, very thick] (2.25,-2.5) -- (6,-2.5) node [near end, label=below :{$I_{y_4}$}] {};
   
   \draw [draw=blue, very thick] (2.25,-3.5) -- (6,-3.5) node [near end, label=below :{$I_{z_4}$}] {};

       \draw [draw=blue, very thick] (11,-1) -- (12.5,-1) node [near end, label=below :{$I_v$}] {};
       
          \draw [draw=blue, very thick] (8.75,-2) -- (12.5,-2) node [near start, label=below :{$I_{x_{k-1}}$}] {};
   
   \draw [draw=blue, very thick] (8.75,-3) -- (12.5,-3) node [near start, label=below :{$I_{y_{k-1}}$}] {};
   
   \draw [draw=blue, very thick] (8.75,-4) -- (12.5,-4) node [near start, label=below :{$I_{z_{k-1}}$}] {};

   \draw (7,-1)  -- (7.75,-1)[dashed];

  \end{tikzpicture}

\section{Conclusions}\label{sec:concl}

While the enumeration of the minimal dominating sets of a graph attracted
a lot of attention, in particular due to its relation to the (output-sensitive) enumeration of the minimal transversals of a hypergraph,
not much is known for the highly related problem of enumeration of the minimal connected dominating sets for a variety of graph classes and for both input- and output-sensitive enumeration. 
First (input-sensitive) enumeration algorithms for \mcds{} of a graph have been obtained in \cite{ch+mcds} which studies among others strongly chordal and AT-free graphs. 
It seems that it is hard to use bipartiteness as a structural aide
when constructing enumeration algorithms. 
In this paper we provided the first enumeration algorithm for some non trivial subclass of  bipartite graphs, namely the convex graphs. 
It would be very interesting to extend our research to enumeration algorithms and upper bounds
on the number of \mcds{} in chordal bipartite graphs and bipartite graphs,
on the number of minimal dominating sets and maximal irredundant sets in convex bipartite graphs on the other hand.


\begin{thebibliography}{10}
\providecommand{\url}[1]{\texttt{#1}}
\providecommand{\urlprefix}{URL }


\bibitem{Lokshtanovmcds}
Daniel Lokshtanov, Michal Pilipczuk, and Saket Saurabh.
\newblock Below all subsets for minimal connected dominating set.
\newblock {\em {SIAM} J. Discrete Math.}, 32(3):2332--2345, 2018.


\bibitem{ch+mcds}
P.~A. Golovach, P.~Heggernes, D.~Kratsch,
{Enumerating minimal
  connected dominating sets in graphs of bounded chordality}, Theor. Comput.
  Sci. 630 (2016) 63--75.
  
  
\bibitem{Abu-Khzam:2011:EAC:1994013.1994228}
F.~N. Abu-Khzam, A.~E. Mouawad, M.~Liedloff, An exact algorithm for connected
  red-blue dominating set, J. of Discrete Algorithms 9~(3) (2011) 252--262.
  
  
  \bibitem{sp+mcds}
Skjørten, I.B.: Faster enumeration of minimal connected dominating sets in split graphs (2017)


  \bibitem{DamaschkeMK90}
Damaschke, P., M{\"{u}}ller, H., Kratsch, D.: Domination in convex and chordal
  bipartite graphs. Inf. Process. Lett.  36(5),  231--236 (1990),
  \url{https://doi.org/10.1016/0020-0190(90)90147-P}
  
  
  \bibitem{jensen}
Bang-jensen, J., Huang, J., Macgillivray, G., Yeo, A.: Domination in convex
  bipartite and convex-round graphs  (1999)
  
  
\bibitem{bipartite}
S.T.~Hedetniemi J.~Pfaff, R.~Laskar.
\newblock NP-completeness of total and connected domination and irredundance
  for bipartite graphs.
\newblock {\em Technical Report 428, Department of Mathematical Sciences,
  Clemenson University}, 1983.
  
  
  \bibitem{Haiko}
Haiko Müller and Andreas Brandstädt.
\newblock The NP-completeness of steiner tree and dominating set for chordal
  bipartite graphs.
\newblock {\em Theoretical Computer Science}, 53(2):257 -- 265, 1987.


  \bibitem{FominGPS08}
Fomin, F.V., Grandoni, F., Pyatkin, A.V., Stepanov, A.A.: Combinatorial bounds
  via measure and conquer: Bounding minimal dominating sets and applications.
  {ACM} Trans. Algorithms  5(1),  9:1--9:17 (2008),
  \url{http://doi.acm.org/10.1145/1435375.1435384}
  
  
  \bibitem{CouturierHHK13}
Couturier, J., Heggernes, P., van~'t Hof, P., Kratsch, D.: Minimal dominating
  sets in graph classes: Combinatorial bounds and enumeration. Theor. Comput.
  Sci.  487,  82--94 (2013), \url{https://doi.org/10.1016/j.tcs.2013.03.026}
  
  
  \bibitem{GolovachKLS19}
Petr~A. Golovach, Dieter Kratsch, Mathieu Liedloff, and Mohamed~Yosri Sayadi.
\newblock Enumeration and maximum number of minimal dominating sets for chordal
  graphs.
\newblock {\em Theor. Comput. Sci.}, 783:41--52, 2019.


 \bibitem{KanteLMN14}
Kant{\'{e}}, M.M., Limouzy, V., Mary, A., Nourine, L.: On the enumeration of
  minimal dominating sets and related notions. {SIAM} J. Discrete Math.  28(4),
   1916--1929 (2014), \url{https://doi.org/10.1137/120862612}
   
   
\bibitem{Moon1965}
J.~{Moon}, L.~{Moser}, {On cliques in graphs.}, {Israel J. Math} 3 (1965)
  23--28.
  
  
\bibitem{Garey1979}
M.~R. Garey, D.~S. Johnson, Computers and Intractability: A Guide to the Theory
  of NP-Completeness, W. H. Freeman \& Co., New York, NY, USA, 1979.
  
  
\bibitem{Brandstadt:1999}
A.~Brandst\"{a}dt, V.~B. Le, J.~P. Spinrad, Graph Classes: A Survey, Society
  for Industrial and Applied Mathematics, Philadelphia, PA, USA, 1999.
  
  
 \bibitem{FominK10}
Fomin, F.V., Kratsch, D.: Exact Exponential Algorithms. Texts in Theoretical
  Computer Science. An {EATCS} Series, Springer (2010),
  \url{https://doi.org/10.1007/978-3-642-16533-7}
  
  
 \bibitem{Golovach2018}
Petr~A. Golovach, Dieter Kratsch, and Mohamed~Yosri Sayadi.
\newblock Enumeration of maximal irredundant sets for claw-free graphs.
\newblock {\em Theor. Comput. Sci.}, 754:3--15, 2019.


  \bibitem{GolovachKLS17}
Petr~A. Golovach, Dieter Kratsch, Mathieu Liedloff, and Mohamed~Yosri Sayadi.
\newblock Enumeration and maximum number of maximal irredundant sets for
  chordal graphs.
\newblock {\em Discrete Applied Mathematics}, 265:69 -- 85, 2019.
  
 \end{thebibliography}
\end{document}